\documentclass[runningheads]{llncs}
\usepackage[T1]{fontenc}
\usepackage{tabularx}
\usepackage{graphicx}
\usepackage{times}
\usepackage{soul}
\usepackage{url}
\usepackage{hyperref} 
\usepackage[utf8]{inputenc}
\usepackage[small]{caption}
\usepackage{graphicx}
\usepackage{amsmath}
\DeclareMathOperator*{\argmin}{arg\,min}
\usepackage{booktabs}
\usepackage{algorithm}
\usepackage{algorithm}
\usepackage[switch]{lineno}
\usepackage{makecell} 
\usepackage{amssymb}
\usepackage{mdframed}
\usepackage{algpseudocode}
\usepackage{enumitem}
\usepackage{graphicx}
\usepackage{wrapfig}
\usepackage{array}
\usepackage{color}

\begin{document}
\title{Strategic Infrastructure Design via Multi-Agent Congestion Games with Joint Placement and Pricing}
\titlerunning{Strategic Infrastructure Design}
%
\author{Niloofar Aminikalibar 
\and
Farzaneh Farhadi
\and
Maria Chli}
%
\authorrunning{N. Aminikalibar et al.}
%
\institute{Aston University, Birmingham, UK, B4 7ET
\email{\{namin21,f.farhadi,m.chli\}@aston.ac.uk}}
\maketitle              
\begin{abstract}
Real-world infrastructure planning increasingly involves strategic interactions among autonomous agents competing over congestible, limited resources. Applications such as Electric Vehicle (EV) charging, emergency response, and intelligent transportation require coordinated resource placement and pricing decisions, while anticipating the adaptive behaviour of decentralised, self-interested agents.
We propose a novel multi-agent framework for joint placement and pricing under such interactions, formalised as a bi-level optimisation model. The upper level represents a central planner, while the lower level captures agent responses via coupled non-atomic congestion games. Motivated by the EV charging domain, we study a setting where a central planner provisions chargers and road capacity under budget and profitability constraints.
The agent population includes both EV drivers and non-charging drivers (NCDs), who respond to congestion, delays, and costs. To solve the resulting NP-hard problem, we introduce ABO-MPN, a double-layer approximation framework that decouples agent types, applies integer adjustment and rounding, and targets high-impact placement and pricing decisions. Experiments on benchmark networks show that our model reduces social cost by up to 40\% compared to placement- or pricing-only baselines, and generalises to other MAS-relevant domains.

\end{abstract}

\section{Introduction}

Many real-world planning tasks involve \emph{resource provisioning with strategic agents}, where a central planner allocates congestible resources while anticipating agent responses. Applications span facility location, cloud provisioning, and intelligent transportation. While prior work often optimises siting or pricing in isolation, jointly optimising both under congestion-aware, cost-sensitive behaviour presents greater complexity. Such problems naturally induce Nash Equilibrium (NE) responses from agents.

We focus on a key instance: EV fast-charging infrastructure planning. With EVs central to global decarbonisation~\cite{bibra2023global}, scalable and efficient charging networks are critical. Planning is complicated by the interdependence between station viability and strategic driver behaviour. A unified optimisation framework must model equilibrium interactions and jointly determine pricing and placement.

In this setting, a planner provisions two coupled resources, charging stations and road links, while accounting for heterogeneous agents (EV drivers and NCDs) with distinct Origin-Destination (O-D) pairs and cost structures based on time, delay, and fees. This MAS challenge demands integration of decentralised behaviour with centralised, explainable planning.
We model this as a mixed integer non-linear programming (MINLP), capturing joint placement and pricing under agent equilibrium and proving NP-hardness. We derive a tractable NE reformulation for the coupled congestion game, and introduce ABO-MPN, a scalable two-layer approximation combining behavioural decomposition and integer adjustment. Empirical results show 10-40\% social cost reduction over baselines.

\section{Literature Review} \label{sec:literature}

Research on congestion games has established a foundational framework for modelling strategic interactions over shared resources. Rosenthal’s seminal work \cite{RN178} introduced pure-strategy equilibria in non-atomic congestion games, later extended to capture agent heterogeneity, resource-specific costs, and continuous strategies \cite{milchtaich1996congestion}—key for transportation, cloud computing, and network routing \cite{roughgarden2002bad,wardrop1952road}.
However, most studies focus on single-resource settings such as road networks \cite{RN57}, public facilities \cite{aboolian2023location}, or cloud bandwidth \cite{sun2021price}, optimising either placement \cite{RN57,jung2014stochastic,RN194,RN167} or pricing \cite{RN153,RN140,zhao2021dynamic,xiong2016optimal}, rarely both. When joint siting and pricing are considered, agent responses are often simplified or static (e.g., \cite{RN176,RN155}), without fully modelling congestion-aware equilibrium.
Our work advances this by explicitly modelling two congestible resource types, charging stations and roads, capturing coupled agent behaviour via equilibrium constraints derived from non-atomic congestion games, and solving the resulting bi-level MINLP using a novel double-layer approximation method (ABO-MPN).

From a MAS perspective, this represents progress in equilibrium-aware planning: Previous MAS approaches have used agent-based simulations (ABS)~\cite{sadeghi2023agamas}, Reinforcement Lea-rning (RL)~\cite{RN155}, or local coordination~\cite{joe2021coordinating}, but rarely structured optimisation with embedded equilibrium reasoning, limiting interpretability.
Building on and generalising Network Design Problems~\cite{farahani2013review}, which assumes static traffic, we address joint placement and pricing under dynamic, strategic agent responses. This enables interpretable, policy-driven infrastructure planning across domains such as EV charging, cloud provisioning, and emergency response.
Unlike ABS or black-box RL, our model delivers transparent, equilibrium-aware policies grounded in game theory, facilitating coordinated decision-making among self-interested agents.

\section{Strategic Multi-Agent Model} \label{sec:section2}

We formalise EV infrastructure planning as a multi-agent, congestion-aware optimisation problem in which a central authority allocates congestible resources, charging stations and road links, while anticipating the strategic responses of self-interested agents. The authority determines station placement and pricing but does not intervene in individual routing choices or impose traffic controls. The model integrates discrete infrastructure decisions with the equilibrium behaviour of heterogeneous agents, including EV drivers and NCDs, navigating a shared network.

\subsection{Network and Agent Definitions}\label{sec:Charging Network Architecture}

\noindent \textbf{Network:} The transportation system is modelled as a directed graph \( G = (N, A) \), where \( N \) is the set of nodes and \( A \) the set of links. Each link \( l \in A \) has length \( d_l \) and capacity \( c_l \). For each O-D pair \( \omega \in W \), let \( R_\omega \) denote the set of feasible routes. To incorporate charging decisions, we define the set of \emph{extended paths} \( P_\omega \), where each path \( p = (r, i) \) combines a route \( r \in R_\omega \) with a charging node \( i \in r \). The mapping \( s(p) = i \) specifies the charging location on path \( p \), enabling joint modelling of road and station congestion.

\noindent \textbf{Agents:} There are two agent types:

\noindent \textit{1- Non-Charging Drivers (NCDs):} These include non-EVs and EVs that do not need en-route charging. They select routes \( r \in R_\omega \) to minimise travel time and contribute only to road congestion.

\noindent \textit{2- EV Drivers:} These agents require en-route charging and select extended paths \( p \in P_\omega \) to minimise a cost function including travel time, queuing delay, and charging fees. Their choices affect both road and station congestion.

We assume a non-atomic setting where agents are infinitesimal and act independently, while an atomic variant of this problem is studied in~\cite{niloofar2}. For each \( \omega \in W \), let \( \gamma_\omega \) and \( \gamma^0_\omega \) denote the demand of EV drivers and NCDs, respectively. The flow distributions \( \mathbf{q}_\omega = (q_{\omega,p})_{p \in P_\omega} \) and \( \mathbf{q}^0_\omega = (q^0_{\omega,r})_{r \in R_\omega} \) represent the equilibrium traffic distribution for EV drivers and NCDs respectively.

\noindent \textbf{Planner:} The Government Authority (GA) determines the number of chargers \( x_i \in \mathbb{Z}_{\geq 0} \) and charging price \( y_i \geq 0 \) at each node \( i \in N \). These decisions consider electricity prices \( e_i \), maintenance and rental costs \( T_i \), a total charger budget \( B \), and a profitability constraint requiring that revenue at each node covers operational costs by a factor of \( \pi > 1 \). Full constraints are detailed in Section~\ref{sec:GA-objective}.

\subsection{Cost Structure and Agent Behaviour} \label{sec: behaviour factors}

Each EV driver associated with O-D pair \( \omega \) minimises a weighted sum of travel time, queuing delay, and charging fee. Let \( \lambda_1 \), \( \lambda_2 \), and \( \lambda_3 \) denote the respective weights. The perceived cost for choosing extended path \( p \in P_\omega \) is:

\begin{align}
C_{\omega, p}(\mathbf{Q}) = \lambda_1 F_{\omega, p} + \lambda_2 G_p + \lambda_3 Y_p,
\end{align}
with expected cost:

\begin{align}
C_{\omega}(\mathbf{Q}) = \sum_{p \in P_\omega} q_{\omega, p} \, C_{\omega, p}(\mathbf{Q}). \label{eq-EVcost}
\end{align}

Here, \( F_{\omega, p} \) captures total travel time over path \( p \), derived from congestion-dependent link costs \( f_l = d_l (\gamma_l + \gamma^0_l)/c_l \). The queuing delay \( G_p \) estimates expected wait time at the selected station \( s(p) \), and \( Y_p \) is the flat charging fee at that station.

\noindent NCDs follow the same structure, excluding charging-related terms, with expected cost:

\begin{align}
C^0_{\omega}(\mathbf{Q}) = \sum_{r \in R_\omega} q^0_{\omega, r} \lambda_1 F_{\omega, r}, \label{eq-NonEVcost}
\end{align}

with \( F_{\omega, r} \) defined analogously over route \( r \).

This model follows the assumptions in~\cite{niloofar}, where formulation details are provided.

\subsection{Government Authority’s Optimisation Problem} \label{sec:GA-objective}

The GA acts as a central planner and aims to minimise the total system-wide social cost across all EV drivers and NCDs:\vspace{-10pt}
\begin{equation}
\Theta(\mathbf{Q}) = \sum_{\omega \in W} \gamma_\omega C_\omega(\mathbf{Q}) + \gamma^0_\omega C^0_\omega(\mathbf{Q}).
\end{equation}

To achieve this, it jointly determines the number of chargers $x_i$ and the charging prices $y_i$ at each node $i \in N$, subject to the following constraints:

\noindent\textbf{Budget constraint:} The total number of installed chargers across the network must not exceed the available budget $B$, i.e., $\sum_{i \in N} x_i \leq B$.

\noindent\textbf{Profitability constraint:} At each charging location $i$, the total revenue from charging fees must cover $e_i$ and $T_i$, while also achieving a profit; we consider this profit by $\pi > 1$ coefficient (e.g. $\pi=1.2$ equals 20\% profit margin).

\noindent\textbf{Equilibrium constraints:} The strategy profiles $\mathbf{q}_\omega$ and $\mathbf{q}^0_\omega$ must form an NE under the GA's placement and pricing decisions. That is, no agent can unilaterally reduce their cost by changing strategies. These constraints are non-trivial to derive, as any provisioning policy induces a coupled congestion game among heterogeneous agents. Their equilibrium behaviour depends jointly on road congestion and charging infrastructure conditions. The GA must therefore anticipate these strategic interactions to evaluate its decisions. The formulation of NE constraints is detailed in Section \ref{sec:game-dynamics}.

\noindent\textbf{Feasibility constraints:} Charger counts must be non-negative integers ($x_i \in \mathbb{Z}_{\geq 0}$), prices must be non-negative ($y_i \geq 0$), and demand distributions $q_{\omega,p}$ and $q^0_{\omega,r}$ must lie in $[0,1]$ and sum to one for each O-D pair $\omega$.

\noindent Formally, the GA solves the following bi-level optimisation problem:
\begin{equation*}\label{GA Problem}
\begin{aligned}
&\min_{x,y,\mathbf{Q}} \Theta(\mathbf{Q}) \qquad \qquad \qquad \qquad \qquad\text{\textbf{GA Problem}}\\
  \text{s.t.} \quad &\sum_{i \in N} x_{i}  \leq B,\\ 
      &\pi\left(\sum_{\omega, p: s(p)=i} \gamma_{\omega} q_{\omega,p} e_{i} + x_{i} T_{i}\right) \leq \sum_{\omega, p: s(p)=i} \gamma_{\omega} q_{\omega, p} y_{i}, \quad \forall i,\\
      &\mathbf{q}_\omega \in \argmin_{\mathbf{q}_\omega} C_\omega(\mathbf{Q}), \quad \forall \omega, \quad \text{and } \quad \mathbf{q}^0_\omega \in \argmin_{\mathbf{q}^0_\omega} C^0_\omega(\mathbf{Q}), \quad  \forall \omega,\\
      & \text{Feasibility constraints.}
\end{aligned}
\end{equation*}
This bilevel model captures the interplay between centralised planning and decentralised agent responses, embedding agents’ optimisation into the planner’s objective via implicit equilibrium constraints. To enable tractable reformulation, we next derive explicit conditions representing agent interactions.

\section{Agent Equilibrium Modelling and Reformulation} \label{sec:game-dynamics}

To model decentralised, self-interested agent behaviour, we formulate a coupled congestion game capturing interactions between EV drivers and NCDs. We derive explicit NE conditions to replace the nested optimisation, enabling direct integration of agent responses into the planner’s model and supporting solution approaches.

\subsection{Coupled Congestion Games}

Each driver’s cost (Equations~\eqref{eq-EVcost} and~\eqref{eq-NonEVcost}) depends on congestion across shared resources, road links and, for EVs, charging stations, inducing a system of coupled congestion games that capture agent interactions over the infrastructure.

\noindent\textbf{Road congestion:} All EV drivers and NCDs contribute to traffic on road links. Their routing decisions affect congestion levels and, in turn, travel costs across the network.
    
\noindent\textbf{Charging congestion:} EV drivers also contribute to congestion at charging stations. Their choice of charging location influences queueing delays and financial costs, which in turn affect their overall extended path selection.

For EV drivers, route and charging station choices are jointly made to minimise a combined cost of travel time, queuing delay, and charging fees. This coupling links the two congestion games, as EV drivers’ decisions influence both domains. Consequently, the equilibrium demand reflects strategic behaviour over the integrated transport and charging network. This interdependence is embedded in the GA’s problem through implicit equilibrium constraints, which the next subsection makes explicit for tractability.

\subsection{Nash Equilibrium Conditions} \label{sec:nash}

In non-atomic congestion games, where each agent is infinitesimal and cannot individually impact the system, a demand distribution is at equilibrium if no agent can reduce their cost by unilaterally changing their decision~\cite{roughgarden2002bad}. In our setting, this implies that for each driver type, 
all strategies (i.e., extended paths for EV drivers and routes for NCDs) used with nonzero flow must yield the minimum possible cost.

For EV drivers, the equilibrium condition is:
\begin{equation}\label{eq:NE-EV}
q_{\omega,p} C_{\omega,p}(\mathbf{Q}) \leq q_{\omega,p} C_{\omega,p'}(\mathbf{Q}), \quad \forall \omega \in W, \; \forall p, p' \in P_\omega.
\end{equation}

For NCDs, a similar equilibrium condition holds for all routes in $R_\omega$.

These inequalities ensure demand is allocated only to cost-minimising strategies, jointly characterising the NE of the coupled congestion game. While multiple equilibria may exist, it is well established that in non-atomic congestion games with convex cost functions (with respect to demand distributions), all equilibria induce the same total system cost~\cite{roughgarden2002bad}. Therefore, incorporating any equilibrium profile consistent with Equation~\eqref{eq:NE-EV} suffices to ensure alignment with the GA’s social welfare objective.

This explicit characterisation replaces the original optimisation-based equilibrium constraints with tractable inequalities, enabling their direct integration into the GA’s decision problem. However, as shown in the next subsection, the resulting formulation remains computationally intractable. This motivates the approximation framework introduced in the following section.

\subsection{Computational Complexity of the Reformulated Problem}\label{sec:complexity}

Even with NE constraints explicitly stated (Equation~\eqref{eq:NE-EV}), the reformulated GA problem remains computationally intractable, as formalised in the theorem below.

\begin{theorem}
The GA's optimisation problem, after reformulating the equilibrium constraints as explicit inequalities, is NP-hard.
\end{theorem}

\begin{proof}
The original formulation is a bilevel MINLP, where the upper-level decisions $(x, y)$ determine charging station placement and pricing, and the lower-level captures agent equilibrium distributions $\mathbf{Q}$. Solving bilevel MINLPs is known to be NP-hard, even under linear cost assumptions~\cite{benayed1990}. The reformulated version replaces the lower-level program with equivalent NE constraints. Solving it in polynomial time would imply a polynomial-time solution to the original bilevel problem, which would in turn imply that NP = P, a contradiction to standard complexity-theoretic assumptions.

Moreover, even with fixed pricing and without equilibrium constraints, the problem reduces to a budget-constrained capacitated facility location problem, which is NP-hard \cite{cornuejols1977float}. The full model thus inherits this computational hardness. \qed
\end{proof}

\section{Solution Methodology} \label{sec:method}

Given the NP-hardness of the GA’s reformulated optimisation problem (Section~\ref{sec:complexity}), exact methods become impractical for large networks. To address this, we propose the \textbf{A}gent-\textbf{B}ased \textbf{O}ptimisation of \textbf{M}ulti-\textbf{P}arameter \textbf{N}etworks (\textbf{ABO-MPN}), a double-layer framework that balances tractability and solution quality.

\subsection{Layer 1: Behavioural Decomposition Based on EV Penetration} \label{sec:level1}

We exploit driver population asymmetry to decompose the problem. When EV penetration ($\alpha = \frac{\sum_\omega \gamma_\omega}{\sum_\omega \gamma_\omega + \gamma^0_\omega}$) is low, their impact on congestion is minimal \cite{RN167}. We first solve a simplified GA-NCD subproblem, excluding EV drivers and chargers, to find NCD equilibrium strategies. These serve as fixed background traffic in the GA-EV subproblem, which then optimises charger placement, pricing, and EV behaviour. This decomposition reduces decision variables per stage, enabling more efficient and scalable analysis.

As EV penetration rate increases, their impact on road usage becomes non negligible, and the independence assumption breaks down. In such cases, we apply an iterative refinement procedure: NCD and EV responses are alternately optimised while holding the other fixed, until convergence (Algorithm~\ref{alg:level1}).

\noindent
\begin{minipage}[t]{0.54\textwidth}

\begin{algorithm}[H]
\caption{ABO-MPN (Layer 1): Decomposition of Driver Behaviour Analysis}
\label{alg:level1}
\begin{algorithmic}[1]
\Require EV penetration rate \( \alpha \)
\If{ \( \alpha \) is sufficiently low }
    \State \textbf{Stage 1: Solve GA-NCD subproblem}
    \State Solve GA-NCD to compute NCD demand distribution \( \mathbf{q}^{0}_\omega \)
    \State \textbf{Stage 2: Solve GA-EV subproblem}
    \State Fix \( \mathbf{q}^{0}_\omega \) as exogenous background traffic
    \State Solve GA-EV with fixed background to obtain \( x^*, y^*, \mathbf{q}_\omega \)
\Else
    \State \textbf{Iterative Refinement Procedure}
    \State Initialise \( \mathbf{q}^{0}_\omega \)
    \Repeat
        \State Fix \( \mathbf{q}^{0}_\omega \), solve GA-EV to update \( x, y, \mathbf{q}_\omega \)
        \State Fix \( \mathbf{q}_\omega \), solve GA-NCD to update \( \mathbf{q}^{0}_\omega \)
    \Until{ Convergence to stable equilibrium }
\EndIf
\State \Return \( x^*, y^*, \mathbf{q}_\omega, \mathbf{q}^{0}_\omega \)
\end{algorithmic}
\end{algorithm}
\end{minipage}%
\hfill%
\vrule width 0.2pt%
\hfill%
\begin{minipage}[t]{0.41\textwidth}

\begin{algorithm}[H]
\caption{ABO-MPN (Layer 2): Integer Adjustment of Placement Decisions}
\label{algorithm 2}
\begin{algorithmic}[1]
\State \textbf{Relax} the placement variables: \( x_i \in \mathbb{R}_{\geq 0} \) for all \( i \in N \)
\State \textbf{Solve} the relaxed GA-EV problem, yielding relaxed placement \( \mathbf{x}^* \)
\State \textbf{Sort} fractional parts: \( x^*_i - \lfloor x^*_i \rfloor \) in descending order
\State \textbf{Floor} each value: \( x^{*'}_i = \lfloor x^*_i \rfloor \)
\State \textbf{Initial sum}: \( S_{\text{initial}} = \sum_i x^{*'}_i \)
\State \textbf{Adjustment}: \( \Delta = \sum_i x^*_i - S_{\text{initial}} \)
\For{$i = 1$ to $\Delta$}
    \State Increment largest fractions: \( x^{*'}_i = x^{*'}_i + 1 \)
\EndFor
\State \textbf{Fix adjusted placements:} \( \mathbf{x}^{*'} \)
\State \textbf{Re-solve GA-EV} to optimise pricing \( y \) and EV distribution \( \mathbf{q}_\omega \)

\end{algorithmic}
\end{algorithm}
\end{minipage}

\subsection{Layer 2: Integer Adjustment for Placement Decisions}

The GA-EV subproblem involves integer charger placement variables that limit scalability. To overcome this, we first solve a relaxed version with continuous variables, then apply a budget-preserving rounding algorithm (Algorithm~\ref{algorithm 2}). It sorts and selectively rounds fractional values to maintain feasibility and minimise deviation. A final re-optimisation updates pricing and EV flows based on fixed placements. As shown in Section~\ref{sec:numerical_results}, this approach yields near-optimal designs with minimal error and negligible impact on solution quality.

\section{Experimental Evaluation} \label{sec:numerical_results}

We evaluate the ABO-MPN framework on standard transportation networks, focusing on solution quality and approximation accuracy. The results demonstrate that ABO-MPN delivers near-optimal infrastructure policies with substantial social cost reductions, while maintaining computational tractability.
The framework is implemented in Python using Pyomo and solved with the IPOPT nonlinear solver. While multiple networks are considered, detailed results are presented for the widely used Nguyen-Dupuis (ND) network~\cite{N_D}, a common benchmark in EV infrastructure planning~\cite {RN194,fathollahi2023optimal,ferro2024planning}.

\begin{figure}[t]
    \centering
    \begin{minipage}{0.45\linewidth}
        \centering
        \includegraphics[width=\linewidth]{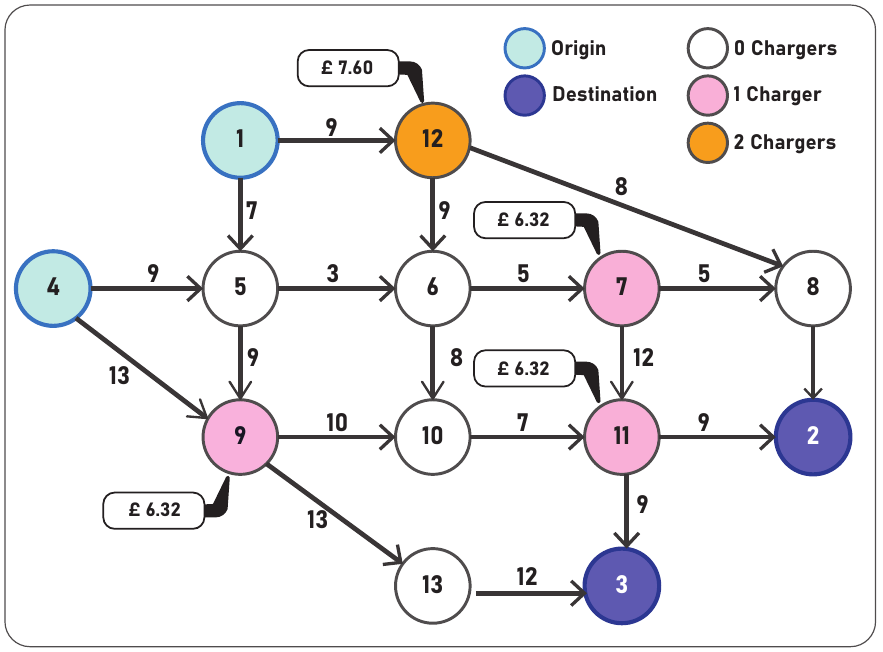}
        \caption{\textbf{Nguyen-Dupuis Network:} \small Optimal placement and pricing in the network}
        \label{fig:5}
    \end{minipage}
    \hfill
    \begin{minipage}{0.54\linewidth}
        \centering
        \begin{table}[H] 
        \caption{ABO-MPN Layer 2 Accuracy Evaluation}
        \label{tab:objective_values_comparison}
        \renewcommand{\arraystretch}{1.2}
        \begin{small}
        \begin{tabular}{|l|c|c|c|c|}
        \hline
        \textbf{Metric} & 
        \makecell{$\lambda_2 = 0.5$\\$B \geq 7$} & 
        \makecell{$\lambda_2 = 2$\\$B \geq 7$} & 
        \makecell{$\lambda_2 = 4$\\$B \geq 7$} & 
        \makecell{$\lambda_2 = 2$\\$B = 3$} \\
        \hline
        \makecell{Social Cost\\(Adjusted $x$)} & 5388.64 & 5605.75 & 5819.28 & 5841.38 \\
        \hline
        \makecell{Social Cost\\(Relaxed $x$)}  & 5349.64 & 5604.21 & 5815.10 & 5787.46 \\
        \hline
        \makecell{Difference\\(\%)}            & 0.72    & 0.02    & 0.07    & 0.90 \\
        \hline
        \end{tabular}
        \end{small}
        \end{table}
    \end{minipage}
\end{figure}

\subsection{Benchmark Network and Parameters} \label{n-d network}

The ND network consists of 13 nodes, 19 links, and 4 O-D pairs: $W = \{(1, 2), (1, 3)$, \\$(4, 2), (4, 3)\}$ (Figure~\ref{fig:5}). Following the SUMO-based study~\cite{RN194}, we use 10 common routes and 54 extended paths. Key parameters are: $(\lambda_1 ,\lambda_2, \lambda_3) = (1,2,3)$, $c_l = 200$, $\mu = 4$, $T_i = 10$, and $\alpha = 13\%$, with 460 driver agents ($\gamma_\omega=15, \gamma^0_\omega=100$ for all O-D pairs). Electricity prices $e_i$ are drawn from $[5, 7]$ for non-O-D nodes and $[11, 15]$ for O-D nodes to reflect urban grid stress.

\subsection{Evaluating ABO-MPN on a Standard Test Case} \label{sec:technique}

We solve the GA problem using ABO-MPN to optimise charger placement, pricing, and EV behaviour. Figure~\ref{fig:Sensitivity on B} (left) shows social cost declines with increasing budget, but plateaus at $B = 7$, beyond which additional chargers yield no further benefit. This plateau occurs because adding chargers beyond $B = 7$ results in over-installation in low-demand areas, increasing costs without improving access or reducing queues. The optimal policy at $B \geq 7$ installs chargers at nodes 7, 9, 11, and 12. Figure~\ref{fig:5} illustrates the resulting network configuration.

\begin{figure}[t]
    \centering
    \includegraphics[width=0.9\linewidth]{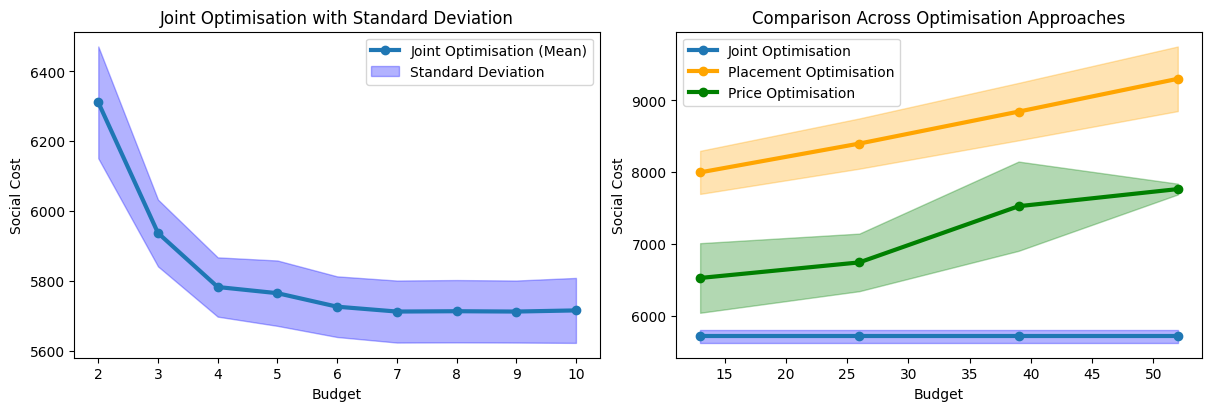}
    \caption{\small \textit{Left:} Social cost vs.\ budget for the joint model, showing diminishing returns beyond $B{=}7$, with shaded regions indicating standard deviation over random parameters across $e$ values. \textit{Right:} Joint Optimisation outperforms benchmarks by at least 10\% across all budgets, with over 40\% gain at $B{=}52$, highlighting the value of joint placement and pricing.}
    \label{fig:Sensitivity on B}
\end{figure}

\subsection{Comparison with State-of-the-Art Methods} \label{sec:comparison}

We benchmark our Joint Optimisation framework against two baselines: \textit{Price Optimisation}, which uses uniform charger placement and optimises pricing, and \textit{Placement Optimisation}, which fixes pricing and optimises locations—following prior works (Section~\ref{sec:literature}).
As shown in Figure~\ref{fig:Sensitivity on B} (right), Joint Optimisation consistently outperforms both, with social cost improvements starting at 10\% and exceeding 40\% at $B = 52$. By aligning placement and pricing with demand and local costs, it avoids inefficiencies from over-installation or mispricing, enabling more effective infrastructure deployment.

\subsection{Evaluation of ABO-MPN Approximation Accuracy} \label{sec: AA}
We evaluate the accuracy of each level in the ABO-MPN approximation framework:

\noindent\textbf{Layer 1: Behavioural Decomposition.} We validate the independence assumption by comparing NCD equilibrium flows in decomposed and joint models. The observed deviation is less than 1\%, justifying the simplification under low EV penetration rates.

\noindent\textbf{Layer 2: Integer Adjustment.} We first solve the GA-EV problem with relaxed (continuous) placement variables \( x \), then apply Algorithm~\ref{algorithm 2} to obtain integer solutions within budget. Table~\ref{tab:objective_values_comparison} shows that the social cost deviation between relaxed and integer solutions stays below 1\%, with a maximum of 0.9\%. As the relaxed solution upper-bounds the true optimum, this yields a conservative estimate of the framework’s optimality gap.

\section{Conclusion}\label{conclusion}

This work tackles the strategic planning of interdependent, congestible infrastructure in settings with self-interested, congestion-aware agents. Focusing on EV fast-charging, we propose a joint optimisation framework for station placement and pricing, anticipating agent responses via equilibrium modelling.
We formulate the problem as a bilevel MINLP, with agent behaviour captured by a coupled non-atomic congestion game. To manage its complexity, we introduce ABO-MPN, a two-layer approximation method combining behavioural decomposition and integer relaxation.

Results show that joint optimisation consistently outperforms pricing- or placement-only baselines, reducing social cost by over 10\% and up to 40\% in high-budget scenarios. The framework ensures operator profitability, enhances public utility, and is robust to behavioural and network variations.

From a MAS perspective, this work contributes a generalisable framework for coordination -aware resource planning, integrating decentralised agent behaviour with centralised, interpretable optimisation. Future work could extend to dynamic pricing, grid constraints, or competitive multi-provider settings.

\section{Acknowledgement}
This paper has been accepted for publication in the Proceedings of the 22nd European Conference on Multi-Agent Systems (EUMAS 2025). The final authenticated version will be published by Springer on SpringerLink.

\end{document}